\documentclass[11pt,a4paper]{article}

\usepackage[utf8]{inputenc}
\usepackage[T1]{fontenc}
\usepackage{lmodern}
\usepackage{amsmath,amssymb,amsthm,graphicx}
\usepackage{bm}
\usepackage[margin=1in]{geometry} 
\usepackage[colorlinks=true, allcolors=blue]{hyperref}

\newtheorem{lemma}{Lemma}
\newtheorem{theorem}{Theorem}

\begin{document}

\title{Classical simulability of quantum circuits followed by\\ sparse classical post-processing}

\author{Yasuhiro Takahashi\textsuperscript{1, 2, }\thanks{\tt ya.takahashi@cs.tsukuba.ac.jp} \hspace{0.8cm} Masayuki Miyamoto\textsuperscript{1, 2} \hspace{0.8cm} Noboru Kunihiro\textsuperscript{1, 2}\\
\textsuperscript{1} Institute of Systems and Information Engineering, University of Tsukuba,\\
Tsukuba, Ibaraki 305-8573, Japan\\
\textsuperscript{2} Tsukuba Institute for Advanced Research, University of Tsukuba,\\
Tsukuba, Ibaraki 305-8577, Japan}
\date{}
\maketitle
\vspace{-1em}

\begin{abstract}
We study the classical simulability of a polynomial-size quantum circuit $C_n$ on $n$ qubits followed by sparse 
classical post-processing (SCP) on $m$ bits, where $m \leq n \leq {\rm poly}(m)$. The SCP is described by a 
non-zero Boolean function $f_m$ that is classically computable in polynomial time and is sparse, i.e., has a peaked 
Fourier spectrum. First, we provide a necessary and sufficient condition on $C_n$ such that, for any SCP $f_m$, $C_n$ followed 
by $f_m$ is classically simulable. This characterization extends the result of Van den Nest and implies that various quantum 
circuits followed by SCP are classically simulable. Examples include IQP circuits, Clifford Magic circuits, and the quantum part of Simon's algorithm, 
even though these circuits alone are hard to simulate classically. 
Then, we consider the case where $C_n$ has constant depth~$d$. While it is unlikely that, for any SCP $f_m$, 
$C_n$ followed by $f_m$ is classically simulable, we show that it is simulable by a polynomial-time probabilistic algorithm with access 
to commuting quantum circuits on $n+1$ qubits. Each such circuit consists 
of at most $\deg(f_m)$ commuting gates and each commuting gate acts on at most $2^d+1$ qubits, 
where $\deg(f_m)$ is the Fourier degree of $f_m$. This provides a better understanding of the hardness 
of simulating constant-depth quantum circuits followed by SCP.
\end{abstract}


\section{Introduction}

\subsection{Background and main results}
A central challenge in quantum computation is to clarify the computational gap between quantum and 
classical computers~\cite{Mon}. One fruitful approach toward this goal is to study the classical simulability of 
quantum computational models~\cite{Harrow}. In particular, since well-known fast quantum algorithms, such as Simon's and Shor's algorithms~\cite{Simon,Shor}, rely on 
classical post-processing, it is natural to examine a model based on a quantum circuit followed by classical post-processing. 
Van den Nest studied Simon-type quantum circuits followed by sparse classical post-processing (SCP), 
which simplifies the non-sparse post-processing in the original algorithm, and showed that such computations are 
classically simulable~\cite{van2}. Here, classical post-processing is said to be sparse if the Boolean function describing it 
has a peaked Fourier spectrum. However, as we discuss in Section 3, constant-depth quantum circuits followed by 
SCP are unlikely to be classically simulable. This sharp contrast with the Simon-type setting suggests that the classical 
simulability of quantum circuits followed by SCP is highly sensitive to the structure of the underlying quantum circuit. 
Therefore, to clarify the boundary between classically simulable quantum computations and those 
that are not, it is essential to investigate the classical simulability of general quantum circuits followed by SCP.

In this paper, we study the classical simulability of a computational model consisting of a polynomial-size quantum circuit $C_n$ on $n$ qubits 
and a non-zero Boolean function $f_m:\{0,1\}^m \to \{0,1\}$, which is an SCP, where $m \leq n \leq {\rm poly}(m)$. Specifically, $C_n$ is applied to the initial state $|0^n\rangle$, 
$m$ qubits are measured in the computational basis, and the classical outcomes of the measurements are processed by $f_m$. We assume that $f_m$ is classically computable in polynomial time and is sparse, 
i.e., has a peaked Fourier spectrum in the sense that the Fourier sparsity of 
$f_m$ is upper-bounded by some known polynomial in $m$: 
$$| \{s\in \{0,1\}^m \mid \widehat{f_m}(s)\neq 0\}| \leq {\rm poly}(m),$$
where the Fourier coefficients are given by
$$\widehat{f_m}(s)=\frac{1}{2^m}\sum_{x \in\{0,1\}^m}f_m(x)(-1)^{s\cdot x}.$$
The overall computation, $C_n$ followed by $f_m$, is illustrated in Fig.~\ref{figure1}. This computation is said to be classically simulable if there exists a 
polynomial-time probabilistic algorithm that, with exponentially-small error probability, approximates the probability that the output of the computation is~1 
with polynomially-small additive error in $n$, given the classical description of $C_n$ and a polynomial-time classical algorithm for evaluating $f_m$~\cite{BSS,van2}. 
The same notion of simulability is used throughout this paper for other, more general computational models.

There have been various studies on quantum circuits followed by classical post-processing. To our knowledge, Van den Nest was the first 
to consider quantum circuits followed by SCP~\cite{van2}; however, that work was restricted to quantum circuits with a specific structure similar 
to the quantum part of Simon's algorithm. Bravyi et al.\ investigated a related model~\cite{BSS}, but their setting did not involve SCP. 
In addition to these studies, several works have examined quantum circuits followed by restricted forms of classical post-processing. Slote 
studied the model ${\sf AC}^0 \circ {\sf QNC}^0$~\cite{Slote}, consisting of a ${\sf QNC}^0$ circuit followed by an ${\sf AC}^0$ function. Here, 
a ${\sf QNC}^0$ circuit is a constant-depth quantum circuit and an ${\sf AC}^0$ function is a Boolean function computed by a constant-depth 
classical circuit with unbounded fan-in. This line of research was extended by Anshu et al.~\cite{Anshu} and Ghosh et al.~\cite{Ghosh}. 
The model studied by Takahashi et al.\ can be viewed as a ${\sf QNC}^0$ 
circuit followed by an AND function on an arbitrary number of inputs~\cite{Takahashi-hardness}, although classical post-processing was not explicitly mentioned. 
Despite these efforts, no prior work has analyzed the classical simulability of general quantum circuits followed by SCP, 
which is the focus of the present paper.

\begin{figure}[t]
\centering
\includegraphics[width=0.65\linewidth]{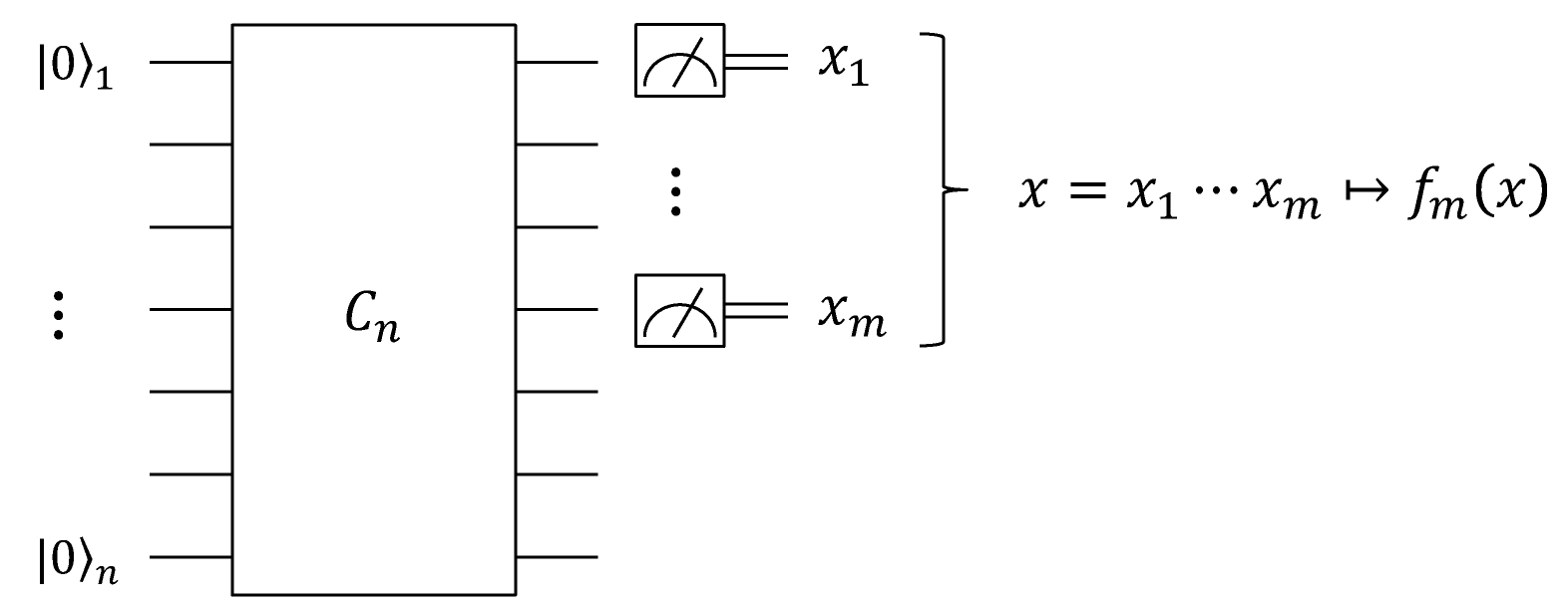}
\caption{$C_n$ followed by $f_m$. The circuit $C_n$ is applied to the initial state $|0^n\rangle$, the top $m$ qubits are measured in the computational basis to obtain a classical outcome $x= x_1\cdots x_m \in \{0,1\}^m$, and the non-zero Boolean function $f_m$ is applied to $x$, yielding the output $f_m(x) \in \{0,1\}$.}
\label{figure1}
\end{figure}

First, we provide a necessary and sufficient condition on $C_n$ such that, for any SCP $f_m$, $C_n$ followed by $f_m$ is classically simulable. 
This is done by using the classical approximability of certain real numbers defined from $C_n$, 
where ``classically approximable'' is understood in the same sense as in the above definition of classical simulability, 
namely that the quantity can be approximated by a polynomial-time probabilistic algorithm with polynomially-small additive error and with exponentially-small error 
probability. For any $s=s_1\cdots s_m \in \{0,1\}^m$, we define
$$Z(s) = \bigotimes_{j=1}^m Z_j^{s_j},$$
which is the unitary operation on the $m$ measured qubits. Here, $Z_j^0=I$ and $Z_j^1=Z_j$, where $Z_j$ denotes the 
Pauli-$Z$ operation on the $j$-th qubit. Let $I_{n-m}$ be the identity operation on the $n-m$ unmeasured qubits. Our first result is described as follows:

\begin{theorem}\label{main-theorem}
Let $C_n$ be an arbitrary polynomial-size quantum circuit on $n$ qubits initialized to $|0^n\rangle$, and $m$ be an arbitrary integer such that 
$m \leq n \leq {\rm poly}(m)$. The following items are equivalent:
\begin{itemize}
\item[\rm (1)] For any SCP $f_m$, $C_n$ followed by $f_m$ is classically simulable.

\item[\rm (2)] For any $s \in \{0,1\}^m \setminus\{0^m\}$, the Pauli expectation value
$$\langle 0^n |C_n^\dagger (Z(s) \otimes I_{n-m}) C_n |0^n\rangle$$
is classically approximable.
\end{itemize}
\end{theorem}
\noindent
Theorem~\ref{main-theorem} implies that various quantum circuits followed by SCP are classically simulable. 
A primary example is the case where $C_n$ has a structure similar to the quantum part of Simon's 
algorithm. This includes IQP circuits~\cite{Shepherd}, which are regarded as candidates for demonstrating quantum supremacy~\cite{Harrow} 
since classically sampling from their output probability distributions is believed to be computationally intractable. By leveraging the properties of CT 
states and ECS operations~\cite{van2}, we can show that item (2) in Theorem~\ref{main-theorem} holds for these circuits. Consequently, 
by Theorem~\ref{main-theorem}, for any SCP $f_m$, $C_n$ followed by $f_m$ is classically simulable. This recovers the result 
of Van den Nest, while explicitly clarifying the applicability to IQP circuits, which was not formally 
addressed in~\cite{van2}. Furthermore, Theorem~\ref{main-theorem} readily applies to the case where $C_n$ is a Clifford Magic circuit~\cite{Yoganathan}, another 
prominent model for quantum supremacy. As in the previous example, item (2) holds and thus, for any SCP $f_m$, 
$C_n$ followed by $f_m$ is classically simulable. Notably, the original argument in~\cite{van2} does not directly imply this 
result for Clifford Magic circuits and thus Theorem~\ref{main-theorem} constitutes a proper extension of the previous work.

Then, we consider the case where $C_n$ is a constant-depth quantum circuit, another model associated with quantum supremacy~\cite{Harrow}. 
As we discuss in Section 3, it is unlikely that, for any SCP $f_m$, $C_n$ followed by $f_m$ is classically simulable. However, we show that 
it becomes simulable if the polynomial-time probabilistic algorithm is augmented with access to commuting quantum circuits. To describe this, 
for any non-zero real-valued Boolean function $h_m$ on $m$ bits, we define its Fourier degree $\deg(h_m)$ as
$$\max \{|s| \mid \widehat{h_m}(s)\neq 0, s \in \{0,1\}^m \},$$
where $|s|$ denotes the Hamming weight of $s$. Moreover, for any unitary operation $U$, we define its support ${\rm supp}(U)$ as 
the set of qubits on which $U$ acts non-trivially. Our second result is described as follows:
\begin{theorem}\label{second-theorem}
Let $C_n$ be an arbitrary polynomial-size quantum circuit of constant depth $d$ on $n$ qubits initialized to $|0^n\rangle$, 
and $m$ be an arbitrary integer such that $m \leq n \leq {\rm poly}(m)$. Then, for any SCP $f_m$, $C_n$ followed by $f_m$ is 
simulable by a polynomial-time probabilistic algorithm with access to commuting quantum circuits on $n+1$ qubits satisfying the following conditions:
\begin{itemize}
\item The number of commuting gates in each commuting quantum circuit is at most $\deg(f_m)$.

\item The number of qubits on which each commuting gate acts is at most
$$1+ \max \{|{\rm supp}(C_n^\dagger (Z_j \otimes I_{n-m}) C_n)| \mid 1 \leq j \leq m \}.$$
\end{itemize}
\end{theorem}
\noindent
Since $\deg(f_m) \leq m$ and the maximum support size in the second condition is at most $2^d$, each commuting quantum circuit appearing 
in Theorem~\ref{second-theorem} consists of at most $m$ commuting gates. Moreover, these circuits are $(2^d+1)$-local in the sense that each commuting gate acts on at most $2^d+1$ qubits. 
Theorem~\ref{second-theorem} demonstrates, for the first time, that commuting quantum circuits are sufficient for simulating constant-depth 
quantum circuits followed by SCP. Consequently, it establishes a non-trivial upper bound on the quantum resources required 
to simulate this computational model, thereby providing a clearer understanding of the hardness of simulating constant-depth quantum circuits followed by SCP.

\subsection{Overview of techniques}

The implication (1) $\Rightarrow$ (2) in Theorem~\ref{main-theorem} follows straightforwardly by considering $f_m$ as a parity function, 
which isolates a single Fourier coefficient. Thus, the core of the proof lies in showing (2) $\Rightarrow$ (1). Under assumption (2), for an arbitrary SCP $f_m$, 
we construct a polynomial-time probabilistic algorithm to approximate the probability that $C_n$ followed by $f_m$ outputs 1. This probability, 
which is denoted as $p(C_n,f_m)$, can be represented as
$$p(C_n,f_m) = \sum_{x\in \{0,1\}^m} f_m(x)p_m(x),$$
where $p_m$ is the output probability distribution of $C_n$ over $\{0,1\}^m$. A key step is transforming this probability into the Fourier domain. 
Specifically, by applying the Plancherel Theorem~\cite{Odonnel-book} and a general relationship between a quantum circuit and the Fourier coefficients 
of its output probability distribution, which is a slightly extended version of the one shown in~\cite{Takahashi-tcs}, we show that approximating $p(C_n,f_m)$ reduces 
to approximating the value
$$p'(C_n,f_m) = \sum_{s\in \{0,1\}^m} \widehat{g_m}(s) \langle 0^n |C_n^\dagger (Z(s) \otimes I_{n-m}) C_n |0^n\rangle,$$
where $g_m(x)=(-1)^{f_m(x)}$ for any $x\in\{0,1\}^m$. Since $f_m$ is sparse, $g_m$ is also sparse, meaning that $p'(C_n,f_m)$ can be represented as
$$\sum_{s\in L} \widehat{g_m}(s) \langle 0^n |C_n^\dagger (Z(s) \otimes I_{n-m}) C_n |0^n\rangle,$$
where
$$L= \{s\in \{0,1\}^m \mid \widehat{g_m}(s)\neq 0\}$$
and $|L|$ is upper-bounded by some known polynomial in $m$. 
To compute the sum, we employ the Kushilevitz-Mansour algorithm~\cite{Kushilevitz} to efficiently identify the significant Fourier coefficients of $g_m$.  
We then approximate each $\widehat{g_m}(s)$ using the Chernoff-Hoeffding bound and each Pauli expectation value via assumption (2). 
Since the number of terms $|L|$ is polynomially bounded, the total additive error remains polynomially small.

To show Theorem~\ref{second-theorem}, where $C_n$ is a quantum circuit of constant depth $d$ on $n$ qubits, we approximate $p(C_n,f_m)$ by utilizing
the aforementioned proof of (2) $\Rightarrow$ (1). Specifically, it suffices to approximate each Pauli expectation value
$$\langle 0^n |C_n^\dagger (Z(s) \otimes I_{n-m}) C_n |0^n\rangle,$$
but in this setting we instead use a polynomial-time probabilistic algorithm with access to commuting quantum circuits on $n+1$ qubits satisfying the conditions 
on the number of commuting gates and the number of qubits on which each commuting gate acts. Indeed, we adapt the construction 
of Ni et al.~\cite{Ni} for approximating $|\langle 0^n | C_n |0^n\rangle|^2$. More concretely, 
we utilize the decomposition
$$\langle 0^n |C_n^\dagger (Z(s) \otimes I_{n-m}) C_n |0^n\rangle =\langle 0^n | \prod_{1 \leq j \leq m:s_j=1} (C_n^\dagger (Z_j \otimes I_{n-m}) C_n) |0^n\rangle.$$
Each operation $V_j = C_n^\dagger (Z_j \otimes I_{n-m}) C_n$ can be regarded 
as a gate in a commuting quantum circuit, since all $\{V_j\}$ mutually commute. By incorporating these gates into a Hadamard test with 
one ancillary qubit, we can construct a commuting quantum circuit on $n+1$ qubits. The probability of obtaining $0$ by measuring the ancillary qubit is
$$\frac{1+ \langle 0^n | \prod_{j :s_j=1} (C_n^\dagger (Z_j \otimes I_{n-m}) C_n) |0^n\rangle }{2}.$$
By applying the Chernoff-Hoeffding bound, each Pauli expectation value can be approximated efficiently with access to the commuting quantum circuit 
on $n+1$ qubits. This construction directly satisfies the conditions of Theorem~\ref{second-theorem}: the number of commuting gates in the commuting 
quantum circuit is $|s|$, which is upper-bounded by $\deg(g_m) = \deg(f_m)$, and the number of qubits on which each commuting gate acts is 
at most
$$1+ \max \{|{\rm supp}(C_n^\dagger (Z_j \otimes I_{n-m}) C_n)| \mid 1 \leq j \leq m \}.$$
Combining these ingredients yields the simulation claimed in Theorem~\ref{second-theorem}.

\section{Preliminaries}

\subsection{Quantum circuits}

We use the standard notation for quantum states and the standard diagrams for quantum circuits~\cite{Nielsen}. A quantum circuit consists of elementary gates, 
each of which is in the gate set $\{H,T,CZ\}$, where
$$
H=\frac{1}{\sqrt{2}}
\begin{pmatrix}
1 & 1 \\
1 & -1
\end{pmatrix},\
T=
\begin{pmatrix}
1 & 0 \\
0 & e^{\frac{\pi i}{4}}
\end{pmatrix},\
CZ=
\begin{pmatrix}
1 & 0 & 0 & 0\\
0 & 1 & 0 & 0\\
0 & 0 & 1 & 0\\
0 & 0 & 0 & -1
\end{pmatrix}.
$$
This gate set is approximately universal for quantum computation. We define $S=T^2$, $Z=T^4$, and $X=HZH$. 
We sometimes use $CCZ$ gates, where
$$CCZ=|00\rangle\langle 00|\otimes I + |01\rangle\langle 01|\otimes I +|10\rangle\langle 10|\otimes I + |11\rangle\langle 11|\otimes Z,$$
which can be decomposed exactly into a constant number of the above elementary gates. For any unitary operation $U$, we define the support of $U$, 
denoted by ${\rm supp}(U)$, as the set of qubits on which $U$ acts non-trivially.

The complexity measures of a quantum circuit are its size and depth. The size is the number of 
elementary gates in the circuit. To define the depth, 
we regard the circuit as a set of layers $1,\ldots,d$ consisting of elementary gates, where gates 
in the same layer act on pairwise disjoint sets of qubits and any gate in layer $j$ is applied 
before any gate in layer $j+1$. The depth is the smallest possible value of $d$~\cite{Fenner}. Each layer $j$ consists of gates that can 
be performed in parallel, since they act on pairwise disjoint sets of qubits.

We deal with a uniform family of polynomial-size quantum circuits $\{C_n\}_{n\geq 1}$, where each $C_n$ has $n$ input qubits initialized to $|0^n\rangle$. 
The uniformity means that the function $1^n \mapsto \overline{C_n}$ is computable by a polynomial-time classical Turing machine, 
where $\overline{C_n}$ is the classical description of $C_n$~\cite{Nishimura-tcs}. After performing $C_n$, we perform measurements 
of $m$ qubits in the computational basis (i.e., the $Z$-basis), where $m\leq n$. Without loss of generality, 
we assume that the measured qubits are the first $m$ qubits corresponding to the top $m$ qubits in a circuit diagram. 
The output probability distribution $p_m$ of $C_n$ over $\{0,1\}^m$ is defined as
$$p_m(x)=\langle 0^n |C_n ^\dagger (|x\rangle \langle x| \otimes I _{n-m} ) C_n|0^n\rangle,$$ 
where a symbol denoting a quantum circuit also denotes its matrix representation in the $Z$ basis and $I_{n-m}$ is the identity operation 
on the $n-m$ unmeasured qubits.  For any $s=s_1\cdots s_m \in \{0,1\}^m$, we define
$$Z(s) = \bigotimes_{j=1}^m Z_j^{s_j},$$
which is the unitary operation on the $m$ measured qubits. Here, $Z_j^0=I$ and $Z_j^1=Z_j$, where $Z_j$ denotes the 
Pauli-$Z$ operation on the $j$-th qubit.

\subsection{Fourier expansions of Boolean functions}\label{fourier}

Let $f_m:\{0,1\}^m \to {\mathbb R}$ be an arbitrary real-valued Boolean function on $m$ bits. Then, $f_m$ can be 
uniquely represented as an ${\mathbb R}$-linear combination of $2^m$ basis functions
$$f_m(x)=\sum_{s \in\{0,1\}^m}\widehat{f_m}(s)(-1)^{s\cdot x},$$
which is called the Fourier expansion of $f_m$~\cite{Odonnel-book}. Here, $\widehat{f_m}(s)$ is called the Fourier 
coefficient of $f_m$ and the symbol ``$\cdot$'' represents the inner product of two $m$-bit strings modulo 2, i.e.,
$$s\cdot x =\bigoplus_{j=1}^m s_jx_j$$
for any $s=s_1\cdots s_m$ and $x=x_1\cdots x_m\in \{0,1\}^m$, where $\oplus$ denotes addition modulo 2. It holds that
$$\widehat{f_m}(s)=\frac{1}{2^m}\sum_{x \in\{0,1\}^m}f_m(x)(-1)^{s\cdot x}$$
for any $s\in\{0,1\}^m$. The Fourier sparsity of $f_m$ is defined as the number of its non-zero Fourier coefficients
$$| \{s\in \{0,1\}^m \,|\, \widehat{f_m}(s)\neq 0\}|.$$
We say that $f_m$ is sparse if its Fourier sparsity is upper-bounded by some known polynomial in $m$
(although the set of $s\in\{0,1\}^m$ such that $\widehat{f_m}(s)\neq 0$ is not necessarily known). Moreover, 
for any non-zero Boolean function $f_m$, we define its Fourier degree $\deg(f_m)$ as
$$\max \{|s| \mid \widehat{f_m}(s)\neq 0, s \in \{0,1\}^m \},$$
where $|s|=\sum_{j=1}^m s_j,$ which is the Hamming weight of $s$. While we consider real-valued functions 
for generality, we primarily focus on $\{0,1\}$-valued functions $f_m: \{0,1\}^m \to \{0,1\}$. A simple example of a 
sparse Boolean function is the parity function
$$f_m(x)=\bigoplus_{j=1}^m x_j.$$
In this case, $\{s\in \{0,1\}^m \,|\, \widehat{f_m}(s)\neq 0\}=\{0^m,1^m\}$ and $\deg(f_m) =m$.

\subsection{Quantum circuits followed by SCP}

Let $C_n$ be an arbitrary polynomial-size quantum circuit on $n$ qubits and $f_m:\{0,1\}^m \to \{0,1\}$ be an arbitrary non-zero sparse Boolean 
function that is classically computable in polynomial time, where $m \leq n \leq {\rm poly}(m)$. We regard $f_m$ as a sparse classical post-processing (SCP) 
and consider the following model of computation, which we call $C_n$ followed by $f_m$:
\begin{enumerate}
\item Perform $C_n$ on $n$ qubits initialized to $|0^n\rangle$.

\item Perform measurements of $m$ qubits of the state $C_n|0^n\rangle$ in the $Z$-basis, where the classical outcome of 
the measurement of the $j$-th qubit is $x_j \in \{0,1\}$ for any $1 \leq j \leq m$.

\item Apply $f_m$ to the classical outcomes $x=x_1\cdots x_m \in\{0,1\}^m$ of the measurements.
\end{enumerate}
The output of this computation is $f_m(x) \in \{0,1\}$.

Let $p_m$ be the output probability distribution of $C_n$ as described above. 
The output of $C_n$ followed by $f_m$ is 0 with probability $\sum_{x: f_m(x)=0}p_m(x)$ and 1 
with probability $\sum_{x: f_m(x)=1}p_m(x)$. Thus, we can represent the probability that the output of $C_n$ followed by $f_m$ is 1 as
$$p(C_n,f_m)= \sum_{x\in \{0,1\}^m} f_m(x)p_m(x).$$
This can be viewed as the expectation value of $f_m$ under the probability distribution $p_m$. The overall computation, 
$C_n$ followed by $f_m$, is said to be classically simulable if there exists a polynomial-time probabilistic algorithm that, with exponentially-small error probability, approximates $p(C_n,f_m)$ with 
polynomially-small additive error in $n$, given the classical description of $C_n$ and a polynomial-time classical algorithm for evaluating $f_m$~\cite{BSS,van2}. 
More concretely, for any polynomial $q(n)$, there exists a 
polynomial-time probabilistic algorithm that, with probability greater than $1-1/\exp(n)$, outputs a real number $A$ such that
$$|p(C_n,f_m) - A| \leq \frac{1}{q(n)}.$$

This definition also applies to real numbers: a real number is said to be classically approximable 
if there exists a polynomial-time probabilistic algorithm that, with exponentially-small error probability, approximates it with polynomially-small additive error. 
In Section~\ref{commuting-section}, we introduce a computational model that extends classical 
computation by granting the probabilistic algorithm access to certain quantum resources. A computation is said to be simulable in this extended model 
if $p(C_n,f_m)$ can be approximated by an algorithm in this model with the same accuracy and success probability as above.

\section{General quantum circuits followed by SCP}

\subsection{Hardness of constant-depth quantum circuits followed by SCP}

As a motivation for the characterization provided in Theorem~\ref{main-theorem}, we first discuss the classical simulability of 
constant-depth quantum circuits followed by SCP. Let $C_n$ be an arbitrary constant-depth polynomial-size quantum circuit on $n$ qubits 
initialized to $|0^n\rangle$. For our purposes, it suffices to consider the case where $m=n$ and $f_n$ is the parity function, i.e.,
$$f_n(x)=\bigoplus_{j=1}^n x_j,$$ 
which is a simple example of a sparse Boolean function.

To classically simulate $C_n$ followed by $f_n$, as described in Section 2.3, we need to approximate the probability
$$p(C_n,f_n) = \sum_{x\in \{0,1\}^n} f_n(x)p_n(x) = \sum_{x\in \{0,1\}^n} \left(\bigoplus_{j=1}^n x_j \right) p_n(x).$$
Since it holds that
$$\bigoplus_{j=1}^n x_j = \frac{1}{2}-\frac{1}{2}(-1)^{|x|}$$
and $p_n(x)= \langle 0^n |C_n ^\dagger |x\rangle \langle x| C_n|0^n\rangle$, the above probability can be represented as follows:
 \begin{align*}
p(C_n,f_n) 
&= \sum_{x\in \{0,1\}^n} \left(\frac{1}{2}-\frac{1}{2}(-1)^{|x|} \right) p_n(x) \\
&= \frac{1}{2} - \frac{1}{2} \sum_{x\in \{0,1\}^n} (-1)^{|x|}  \langle 0^n |C_n ^\dagger |x\rangle \langle x| C_n|0^n\rangle \\
&= \frac{1}{2} - \frac{1}{2} \sum_{x\in \{0,1\}^n} \langle 0^n |C_n ^\dagger  Z(1^n) |x\rangle \langle x| C_n|0^n\rangle 
= \frac{1}{2} - \frac{1}{2} \langle 0^n |C_n ^\dagger  Z(1^n) C_n|0^n\rangle,
\end{align*}
where the last equality is due to $\sum_{x\in \{0,1\}^n}|x\rangle\langle x| = I_n$.

This identity implies that, if $C_n$ followed by $f_n$ is classically simulable, then the Pauli expectation value
$$\langle 0^n |C_n ^\dagger  Z(1^n) C_n|0^n\rangle$$
must be classically approximable. Specifically, there would exist a polynomial-time 
probabilistic algorithm that, with exponentially-small error probability, approximates this value with polynomially-small additive error. 
Although a formal complexity‑theoretic hardness proof is not known, such an approximation is widely believed to be 
intractable when $C_n$ is a general constant‑depth quantum circuit~\cite{Ni}; Bravyi et al.\ show that, under special structural 
assumptions (e.g., geometric constraints), related quantities can be approximated classically in quasi‑polynomial or even 
polynomial time~\cite{Bravyi-mean,Bravyi-peaked}, but those algorithms crucially rely on the assumed structure 
and are not believed to extend to the general case.

\subsection{Proof of \texorpdfstring{(1) $\Rightarrow$ (2)}{(1) => (2)} in Theorem~\ref{main-theorem}}

We first give a representation of the probability $p(C_n,f_m)$, which is useful for showing Theorem~\ref{main-theorem}:

\begin{lemma}\label{prob-representation}
Let $C_n$ be an arbitrary polynomial-size quantum circuit on $n$ qubits initialized to $|0^n\rangle$ and $f_m:\{0,1\}^m \to \{0,1\}$ be 
an arbitrary SCP on $m$ bits, where $0 < m \leq n$. Then, the probability $p(C_n,f_m)$ can be represented as
$$p(C_n,f_m) = \frac{1}{2}-\frac{1}{2}\sum_{s\in \{0,1\}^m} \widehat{g_m}(s) \langle 0^n |C_n^\dagger (Z(s) \otimes I_{n-m}) C_n |0^n\rangle,$$
where $g_m(x)=(-1)^{f_m(x)}$ for any $x\in\{0,1\}^m$.
\end{lemma}
\begin{proof}
Let $p_m$ be the output probability distribution of $C_n$ over $\{0,1\}^m$. A simple calculation implies the following 
representation of ${\mathbb E}_x[g_m(x)p_m(x)]$, which is the Plancherel Theorem~\cite{Odonnel-book}:
\begin{align*}
\frac{1}{2^m} \sum_{x\in \{0,1\}^m} g_m(x)p_m(x) &= \frac{1}{2^m} \sum_{x\in \{0,1\}^m} \left(\sum_{s\in \{0,1\}^m} \widehat{g_m}(s)(-1)^{s\cdot x}\right) \left(\sum_{t\in \{0,1\}^m} \widehat{p_m}(t)(-1)^{t\cdot x}\right) \\
&= \frac{1}{2^m}\sum_{s,t\in \{0,1\}^m} \widehat{g_m}(s)\widehat{p_m}(t)\sum_{x\in \{0,1\}^m}(-1)^{(s+t)\cdot x} 
= \sum_{s\in \{0,1\}^m} \widehat{g_m}(s)\widehat{p_m}(s),
\end{align*}
where $s+t$ is the bitwise addition of $s$ and $t$ modulo 2. Thus, we can represent $p(C_n,f_m)$ as follows: 
\begin{align*}
p(C_n,f_m) 
&= \sum_{x\in \{0,1\}^m} f_m(x)p_m(x) = \sum_{x\in \{0,1\}^m} \left(\frac{1}{2}-\frac{1}{2}(-1)^{f_m(x)} \right) p_m(x)\\
&= \frac{1}{2} - \frac{1}{2} \sum_{x\in \{0,1\}^m} g_m(x)p_m(x) =\frac{1}{2} - \frac{2^m}{2}\sum_{s\in \{0,1\}^m} \widehat{g_m}(s)\widehat{p_m}(s).
\end{align*}
Using the identities $p_m(x) = \langle 0^n |C_n^\dagger (|x\rangle\langle x| \otimes I_{n-m}) C_n |0^n\rangle$
and $\sum_{x\in \{0,1\}^m}|x\rangle\langle x| = I_m$, 
we obtain the following relationship between $C_n$ and $\widehat{p_m}(s)$ for any $s\in \{0,1\}^m$, which is a slightly extended version of the one shown in~\cite{Takahashi-tcs}:
\begin{align*}
\widehat{p_m}(s)
&= \frac{1}{2^m}\sum_{x\in \{0,1\}^m} p_m(x)(-1)^{s\cdot x} = \frac{1}{2^m}\sum_{x\in \{0,1\}^m} \langle 0^n |C_n^\dagger (|x\rangle\langle x| \otimes I_{n-m}) C_n |0^n\rangle (-1)^{s\cdot x} \\
&= \frac{1}{2^m}\sum_{x\in \{0,1\}^m} \langle 0^n |C_n^\dagger (Z(s)|x\rangle\langle x| \otimes I_{n-m}) C_n |0^n\rangle
  = \frac{1}{2^m}\langle 0^n |C_n^\dagger (Z(s) \otimes I_{n-m}) C_n |0^n\rangle.
\end{align*}
Therefore,
$$p(C_n,f_m) = \frac{1}{2} - \frac{2^m}{2}\sum_{s\in \{0,1\}^m} \widehat{g_m}(s)\widehat{p_m}(s) =  \frac{1}{2}-\frac{1}{2}\sum_{s\in \{0,1\}^m} \widehat{g_m}(s) \langle 0^n |C_n^\dagger (Z(s) \otimes I_{n-m}) C_n |0^n\rangle,$$
which is the desired relationship.
\end{proof}

Using Lemma~\ref{prob-representation}, we show that (1) $\Rightarrow$ (2) in Theorem~\ref{main-theorem} as the following lemma:
\begin{lemma}\label{half-main-theorem}
Let $C_n$ be an arbitrary polynomial-size quantum circuit on $n$ qubits initialized to $|0^n\rangle$, and $m$ be an arbitrary integer such that 
$m \leq n \leq {\rm poly}(m)$. We assume that, for any SCP $f_m$, $C_n$ followed by $f_m$ is classically simulable. Then, 
for any $s \in \{0,1\}^m\setminus\{0^m\}$, the Pauli expectation value
$$\langle 0^n |C_n^\dagger (Z(s) \otimes I_{n-m}) C_n |0^n\rangle$$
is classically approximable.
\end{lemma}
\begin{proof}
Fix an arbitrary $s\in \{0,1\}^m\setminus\{0^m\}$. We define a Boolean function $h_m^s:\{0,1\}^m \to \{0,1\}$ as
$$h_m^s(x) = s\cdot x = \frac{1}{2} - \frac{1}{2}(-1)^{s\cdot x},$$
which is classically computable in polynomial time. From the representation of the function,
\begin{equation*}
\widehat{h_m^s}(t)=
\begin{cases}
\frac{1}{2} & \text{if $t=0^m$,} \\ 
-\frac{1}{2} & \text{if $t=s$,} \\
0 & \text{otherwise.}
\end{cases}
\end{equation*}
We can verify this using the following relationship: for any $t\in \{0,1\}^m$,
$$\widehat{h_m^s}(t)= \frac{1}{2^m}\sum_{x \in\{0,1\}^m}h_m^s(x)(-1)^{t\cdot x} 
= \frac{1}{2^{m+1}}\sum_{x \in\{0,1\}^m}(-1)^{t\cdot x} - \frac{1}{2^{m+1}}\sum_{x \in\{0,1\}^m}(-1)^{(s+t)\cdot x}.$$
Thus,
$$| \{t\in \{0,1\}^m \mid \widehat{h_m^s}(t)\neq 0\}| = 2,$$
which implies that $h_m^s$ is sparse.

We can regard $h_m^s$ as an SCP. By the assumption of Lemma~\ref{half-main-theorem}, 
$C_n$ followed by $h_m^s$ is classically simulable, which means that $p(C_n,h_m^s)$ is classically approximable. By Lemma~\ref{prob-representation},

$$p(C_n,h_m^s) = \frac{1}{2}-\frac{1}{2}\sum_{t\in \{0,1\}^m} \widehat{g_m^s}(t) \langle 0^n |C_n^\dagger (Z(t) \otimes I_{n-m}) C_n |0^n\rangle,$$
where $g_m^s(x)=(-1)^{h_m^s(x)} = (-1)^{s\cdot x}$ for any $x\in\{0,1\}^m$. From the representation of $g_m^s$, it holds that
\begin{equation*}
\widehat{g_m^s}(t)=
\begin{cases} 
1 & \text{if $t=s$,} \\
0 & \text{otherwise.}
\end{cases}
\end{equation*}
We can verify this using the following relationship: for any $t \in \{0,1\}^m$,
$$\widehat{g_m^s}(t)= \frac{1}{2^m}\sum_{x \in\{0,1\}^m}g_m^s(x)(-1)^{t\cdot x} = \frac{1}{2^m}\sum_{x \in\{0,1\}^m}(-1)^{(s+t)\cdot x}.$$
Thus,
$$p(C_n,h_m^s)= \frac{1}{2} - \frac{1}{2}\langle 0^n |C_n^\dagger (Z(s) \otimes I_{n-m}) C_n |0^n\rangle.$$
Since $p(C_n,h_m^s)$ is classically approximable, the Pauli expectation value 
$$\langle 0^n |C_n^\dagger (Z(s) \otimes I_{n-m}) C_n |0^n\rangle$$
is also classically approximable. This completes the proof of the lemma.
\end{proof}

\subsection{Proof of \texorpdfstring{(2) $\Rightarrow$ (1)}{(2) => (1)} in Theorem~\ref{main-theorem}}\label{Chernoff-section}

To show that (2) $\Rightarrow$ (1) in Theorem~\ref{main-theorem}, by Lemma~\ref{prob-representation}, it suffices to show the following lemma:
\begin{lemma}\label{key-lemma}
Let $C_n$ be an arbitrary polynomial-size quantum circuit on $n$ qubits initialized to $|0^n\rangle$, and $m$ be an arbitrary integer such that 
$m \leq n \leq {\rm poly}(m)$. We assume that, for any $s \in \{0,1\}^m\setminus\{0^m\}$, the Pauli expectation value 
$$\langle 0^n |C_n^\dagger (Z(s) \otimes I_{n-m}) C_n |0^n\rangle$$
is classically approximable. Then, for any SCP $f_m$,
$$p'(C_n,f_m) = \sum_{s\in \{0,1\}^m} \widehat{g_m}(s) \langle 0^n |C_n^\dagger (Z(s) \otimes I_{n-m}) C_n |0^n\rangle$$
is classically approximable, where $g_m(x)=(-1)^{f_m(x)}$ for any $x\in\{0,1\}^m$.
\end{lemma}

We fix an arbitrary SCP $f_m$ and consider $g_m(x)=(-1)^{f_m(x)}$. We first observe that, for any $s\in \{0,1\}^m$, the value $\widehat{g_m}(s)$ is classically approximable. This is a simple application 
of the Chernoff-Hoeffding bound, which can be described as follows~\cite{Ni}: Let $X_1,\ldots,X_K$ be i.i.d. real-valued random variables 
such that $X_j \in [-1,1]$ and $E={\mathbb E}X_j$ for any $1\leq j \leq K$. Then, for any polynomial $q(m)$,
$${\rm Prob}\left[ \left|  \frac{1}{K}\sum_{j=1}^K X_j -E \right| \leq \frac{1}{q(m)} \right] \geq 1-2e^{-\frac{K}{4q(m)^2}}.$$
Indeed, we fix an arbitrary $s\in \{0,1\}^m$ and define $X_j$ as the real-valued random variable
$$X_j = g_m(x)(-1)^{s\cdot x},$$
where $x \in\{0,1\}^m$ is chosen uniformly at random. It holds that $X_j \in [-1,1]$ and
$$E = {\mathbb E}X_j= \frac{1}{2^m}\sum_{x \in\{0,1\}^m}g_m(x)(-1)^{s\cdot x} = \widehat{g_m}(s).$$
We choose $K={\rm poly}(m)$ sufficiently large and compute $g_m(x)(-1)^{s\cdot x}$ for $K$ independently sampled inputs $x \in\{0,1\}^m$. 
By the Chernoff-Hoeffding bound, $\widehat{g_m}(s)$ is approximated by the average of these values with polynomially-small additive error and with 
exponentially-small error probability. Since $f_m$ is classically computable in polynomial time, so is $g_m$. This can be summarized as the following lemma:
\begin{lemma}\label{Chernoff}
There exists a polynomial-time probabilistic algorithm that, for any $s \in \{0,1\}^m$ and polynomial $q(m)$, outputs a real number $A(s)$ such that,
with probability greater than $1-1/\exp(m)$,
$$|\widehat{g_m}(s) -A(s) | \leq \frac{1}{q(m)}.$$
\end{lemma}
\noindent
We note that, since $m \leq n \leq {\rm poly}(m)$, a similar argument implies that the lemma holds when $q(m)$ and $\exp(m)$ is replaced with $q(n)$ and $\exp(n)$, respectively.

Another observation is that the value $p'(C_n,f_m)$ can be represented as the sum of polynomially-many terms. Indeed, it is straightforward to show 
that $g_m$ is sparse:

\begin{lemma}\label{g-sparse}
The function $g_m$ is sparse, where $g_m(x)=(-1)^{f_m(x)}$ for any $x\in\{0,1\}^m$.
\end{lemma}
\begin{proof}
It holds that, for any $x\in\{0,1\}^m$, $g_m(x) = 1 - 2 f_m(x)$. Thus, for any $s\in \{0,1\}^m$,
\begin{align*}
\widehat{g_m}(s)
&= \frac{1}{2^m}\sum_{x \in\{0,1\}^m}g_m(x)(-1)^{s\cdot x} = \frac{1}{2^m}\sum_{x \in\{0,1\}^m} (1-2f_m(x))(-1)^{s\cdot x} \\
&= \frac{1}{2^{m}}\sum_{x \in\{0,1\}^m}(-1)^{s\cdot x} -2\widehat{f_m}(s).
\end{align*}
Thus,
\begin{equation*}
\widehat{g_m}(s)=
\begin{cases} 
1-2\widehat{f_m}(s) & \text{if $s=0^m$,} \\
-2\widehat{f_m}(s) & \text{otherwise.}
\end{cases}
\end{equation*}
This means that, for any $s\in\{0,1\}^m \setminus \{0^m\}$, $\widehat{g_m}(s) \neq 0$ if and only if $\widehat{f_m}(s) \neq 0$. 
Since
$$| \{s\in \{0,1\}^m \mid \widehat{f_m}(s)\neq 0\}|$$
is upper-bounded by some known polynomial in $m$, so is
$$| \{s\in \{0,1\}^m \mid \widehat{g_m}(s)\neq 0\}|.$$ 
More concretely,
$$| \{s\in \{0,1\}^m \mid \widehat{g_m}(s)\neq 0\}| \leq  | \{s\in \{0,1\}^m \mid \widehat{f_m}(s)\neq 0\}| +1,$$
which completes the proof of the lemma.
\end{proof}

By Lemma~\ref{g-sparse}, $p'(C_n,f_m)$ can be represented as
$$p'(C_n,f_m) = \sum_{s\in L} \widehat{g_m}(s) \langle 0^n |C_n^\dagger (Z(s) \otimes I_{n-m}) C_n |0^n\rangle,$$
where
$$L= \{s\in \{0,1\}^m \mid \widehat{g_m}(s)\neq 0\}$$
and $|L|$ is upper-bounded by some polynomial in $m$, which can be 
represented as $q_L(n)$ since $m \leq n \leq {\rm poly}(m)$. We note that 
$q_L(n)$ is known since, by the definition of SCP, the polynomial that upper bounds the Fourier sparsity of $f_m$ is known 
and the inequality in the proof of Lemma~\ref{g-sparse} holds. A key point is that almost all significant elements of $L$ can 
be obtained efficiently by using (a version of) Kushilevitz-Mansour Theorem:

\begin{theorem}[\cite{Kushilevitz}]\label{k-m-theorem}
There exists a polynomial-time probabilistic algorithm that, for any $h:\{0,1\}^m \to \{-1,1\}$ that is classically computable in polynomial time 
and polynomial $\theta(m)$, 
outputs $\widetilde{L} \subseteq \{0,1\}^m$ such that, with probability greater than $1-1/\exp(m)$, for any $s \in \{0,1\}^m$, 
\begin{itemize}
\item If $s \in \widetilde{L}$, then $|\widehat{h}(s)| > 1/(2\theta(m))$.

\item If $s \notin \widetilde{L}$, then $|\widehat{h}(s)| < 1/\theta(m)$.

\item $|\widetilde{L}| < 4\theta(m)^2$.
\end{itemize}
\end{theorem}
\noindent
We note that, since $m \leq n \leq {\rm poly}(m)$, a similar argument implies that the theorem holds when $\theta(m)$ and $\exp(m)$ is replaced 
with $\theta(n)$ and $\exp(n)$, respectively.

Before proving Lemma~\ref{key-lemma}, under the assumption of the lemma, for an arbitrary 
polynomial $p(n)$, we describe a polynomial-time probabilistic algorithm that, with exponentially-small error probability, approximates $p'(C_n,f_m)$ with additive error $1/p(n)$:

\begin{enumerate}
    \item Apply the Kushilevitz--Mansour algorithm in Theorem~\ref{k-m-theorem} to the function $g_m$ with 
    the threshold polynomial $\theta(n)=3p(n)q_L(n)$ to obtain a set $\widetilde{L}\subseteq\{0,1\}^m$ of significant Fourier coefficients of $g_m$.  

    \item For each $s\in\widetilde{L}$, apply the algorithm in Lemma~\ref{Chernoff} 
    to the index $s$ with the accuracy polynomial $q(n)=24p(n)\theta(n)^2$ to obtain an approximation $A(s)$ of $\widehat{g_m}(s)$.

    \item For each $s\in\widetilde{L}$, if $s= 0^m$, set $B(s)=1$. Otherwise, apply the algorithm assumed in Lemma~\ref{key-lemma} to the index $s$ with 
    the accuracy polynomial $r(n)=12p(n)\theta(n)^2$ to obtain an approximation $B(s)$ of the Pauli expectation value 
    $$\langle 0^n |C_n^\dagger (Z(s) \otimes I_{n-m}) C_n |0^n\rangle.$$

    \item Output the value
    $$\sum_{s\in  \widetilde{L}} A(s)B(s).$$
\end{enumerate}

We show Lemma~\ref{key-lemma} as follows, which completes the proof of Theorem~\ref{main-theorem}:

\begin{proof}[Proof of Lemma~\ref{key-lemma}]

It suffices to show that, with probability greater than $1-1/\exp(n)$, 
$$| p'(C_n,f_m) - \sum_{s\in  \widetilde{L}} A(s)B(s) | \leq \frac{1}{p(n)}.$$
The set $\widetilde{L}$ obtained in Step 1 of the above algorithm satisfies that, 
with probability greater than $1-1/\exp(n)$, for any $s \in \{0,1\}^m$,
\begin{itemize}
\item If $s \in \widetilde{L}$, then $|\widehat{g_m}(s)| > 1/(2\theta(n))$.

\item If $s \notin \widetilde{L}$, then $|\widehat{g_m}(s)| < 1/\theta(n)$.

\item $|\widetilde{L}| < 4\theta(n)^2$.
\end{itemize}
Moreover, for each $s \in \widetilde{L}$, with probability greater than $1-1/\exp(n)$,
$$|\widehat{g_m}(s) - A(s) | \leq \frac{1}{q(n)}.$$
Furthermore, for each $s \in \widetilde{L}$ (including the case where $0^m\in \widetilde{L}$), with probability greater than $1-1/\exp(n)$,
$$| \langle 0^n |C_n^\dagger (Z(s) \otimes I_{n-m}) C_n |0^n\rangle -B(s) | \leq \frac{1}{r(n)}.$$
Hence, with probability greater than $1-1/\exp(n)$, all the above relationships hold simultaneously. In what follows, 
we analyze the error conditioned on this event.

We can bound the following quantity as follows:
\begin{align*}
\left|p'(C_n,f_m)-\sum_{s\in \widetilde{L}} A(s)B(s)\right| 
\leq& \left|p'(C_n,f_m) -\sum_{s\in \widetilde{L}} \widehat{g_m}(s) \langle 0^n |C_n^\dagger (Z(s) \otimes I_{n-m}) C_n |0^n\rangle\right|  \\
&+ \left|\sum_{s\in \widetilde{L}} \widehat{g_m}(s) \langle 0^n |C_n^\dagger (Z(s) \otimes I_{n-m}) C_n |0^n\rangle - \sum_{s\in \widetilde{L}} \widehat{g_m}(s)B(s)\right|\\
&+ \left|\sum_{s\in \widetilde{L}} \widehat{g_m}(s)B(s) - \sum_{s\in  \widetilde{L}} A(s)B(s)\right|.
\end{align*}
The first term of the right hand side of the above inequality is upper-bounded as follows:
\begin{align*}
& \left|p'(C_n,f_m) -\sum_{s\in \widetilde{L}} \widehat{g_m}(s) \langle 0^n |C_n^\dagger (Z(s) \otimes I_{n-m}) C_n |0^n\rangle\right|\\
=& \left|\sum_{s\in L \setminus \widetilde{L}} \widehat{g_m}(s) \langle 0^n |C_n^\dagger (Z(s) \otimes I_{n-m}) C_n |0^n\rangle\right|\\
\leq&  \sum_{s\in L \setminus \widetilde{L}} |\widehat{g_m}(s)| |\langle 0^n |C_n^\dagger (Z(s) \otimes I_{n-m}) C_n |0^n\rangle| < \frac{q_L(n)}{\theta(n)}=\frac{1}{3p(n)}, 
\end{align*}
since $\widetilde{L} \subseteq L$, $|\widehat{g_m}(s)| < 1/\theta(n)$ for any $s\in L \setminus \widetilde{L}$, 
$$|\langle 0^n |C_n^\dagger (Z(s) \otimes I_{n-m}) C_n |0^n\rangle| \leq 1$$
for any $s\in \{0,1\}^m$, and $|L|\leq q_L(n)$. 

The second term is upper-bounded as follows:
\begin{align*}
& \left|\sum_{s\in \widetilde{L}} \widehat{g_m}(s) \langle 0^n |C_n^\dagger (Z(s) \otimes I_{n-m}) C_n |0^n\rangle - \sum_{s\in \widetilde{L}} \widehat{g_m}(s)B(s)\right| \\
\leq& \sum_{s\in \widetilde{L}} |\widehat{g_m}(s)||\langle 0^n |C_n^\dagger (Z(s) \otimes I_{n-m}) C_n |0^n\rangle - B(s)| 
< \frac{4\theta(n)^2}{r(n)} = \frac{1}{3p(n)},
\end{align*}
since $|\widehat{g_m}(s)| \leq 1$ for any $s\in \{0,1\}^m$ and $|\widetilde{L}| < 4\theta(n)^2$.

The third term is upper-bounded as follows:
\begin{align*}
& \left|\sum_{s\in \widetilde{L}} \widehat{g_m}(s)B(s) - \sum_{s\in  \widetilde{L}} A(s)B(s)\right| 
\leq \sum_{s\in \widetilde{L}} |\widehat{g_m}(s) - A(s)||B(s)|  \leq \frac{1}{q(n)} \sum_{s\in \widetilde{L}}|B(s)| \\
\leq& \frac{1}{q(n)} \sum_{s\in \widetilde{L}}  \left( \left|\langle 0^n |C_n^\dagger (Z(s) \otimes I_{n-m}) C_n |0^n\rangle - B(s)\right| + 
\left|\langle 0^n |C_n^\dagger (Z(s) \otimes I_{n-m}) C_n |0^n\rangle\right| \right)\\
\leq& \frac{1}{q(n)} \sum_{s\in \widetilde{L}} \left(\frac{1}{r(n)} + 1\right) < \frac{4\theta(n)^2}{q(n)}\left(\frac{1}{r(n)} + 1\right) 
= \frac{1}{6p(n)}\left(\frac{1}{r(n)} + 1\right) < \frac{1}{3p(n)},
\end{align*}
since we can assume that $1/r(n) < 1$ without loss of generality. Therefore,
$$\left|p'(C_n,f_m) - \sum_{s\in \widetilde{L}} A(s)B(s)\right| < \frac{1}{3p(n)} + \frac{1}{3p(n)} +  \frac{1}{3p(n)} = \frac{1}{p(n)},$$
which is the desired relationship.
\end{proof}

\section{Applications of Theorem~\ref{main-theorem}}

\subsection{CT states and ECS operations}\label{idea}

We apply Theorem~\ref{main-theorem} to various quantum circuits for performing classically-hard tasks. A key step in 
applying Theorem~\ref{main-theorem} is to verify that a polynomial-size quantum circuit $C_n$ satisfies item (2) in Theorem~\ref{main-theorem}, or more concretely, 
satisfies that, for any $s \in \{0,1\}^m\setminus\{0^m\}$, the Pauli expectation value
$$\langle 0^n |C_n^\dagger (Z(s) \otimes I_{n-m}) C_n |0^n\rangle$$
is classically approximable. To do this, we introduce computationally tractable (CT) states and efficiently computable sparse (ECS) operations~\cite{van2}, 
where our notion of ECS is a restricted version tailored to our setting.

Let $|\varphi\rangle$ be an arbitrary quantum state on $n$ qubits and $p_\varphi$ be the probability distribution 
over $\{0,1\}^n$ defined as $p_\varphi(x)=|\langle x |\varphi\rangle|^2$. Then, $|\varphi\rangle$ is CT if the following two conditions are satisfied:
\begin{itemize}
\item The probability distribution $p_\varphi$ is classically samplable in polynomial time.

\item For any $x \in\{0,1\}^n$, the inner product $\langle x |\varphi\rangle$ is classically computable in polynomial time.
\end{itemize}
An example of a CT state is a product state. We note that, for simplicity, we require perfect accuracy in sampling probability distributions and computing 
the inner product, but irrational numbers may be involved and thus it suffices to require exponential accuracy. The same convention applies to the following definition of ECS operations.

Let $U$ be an arbitrary operation on $n$ qubits that is both unitary and Hermitian.  We say that $U$ is sparse if there exists a polynomial $s(n)$ such that,
for any $x \in \{0,1\}^n$, $U|x\rangle$ is a linear combination of at most $s(n)$ 
computational basis states. For a sparse operation $U$ (associated with $s(n)$), we define functions
$$
\beta_j : \{0,1\}^n \to \mathbb{C},\ \gamma_j : \{0,1\}^n \to \{0,1\}^n
$$
for each $1 \le j \le s(n)$ as follows: for any $x \in \{0,1\}^n$, consider the column of $U$ indexed by $x$ and 
traverse this column from top to bottom.
\begin{itemize}
    \item If the $j$-th non-zero entry exists:
    \begin{itemize}
        \item $\beta_j(x)$ is the value of this entry.
        \item $\gamma_j(x)$ is the row index of this entry.
    \end{itemize}
    \item If the $j$-th non-zero entry does not exist:
    \begin{itemize}
        \item $\beta_j(x) = 0.$
        \item $\gamma_j(x) = 0^n.$
    \end{itemize}
\end{itemize}
The sparse operation $U$ is ECS if, for any $1 \leq j \leq s(n)$, the functions $\beta_j$ and $\gamma_j$ are classically computable in polynomial time. In 
particular, an ECS operation with $s(n)=1$ is called efficiently computable basis-preserving.

An efficiently computable basis-preserving operation preserves the class of CT states as follows:
\begin{theorem}[\cite{van2}]\label{van2-1}
Let $|\varphi\rangle$ be an arbitrary CT state on $n$ qubits and $U$ be an arbitrary efficiently computable 
basis-preserving operation on $n$ qubits. Then, $U|\varphi\rangle$ is CT.
\end{theorem}\noindent
The following theorem is a specialized version of the result in~\cite{van2}, tailored to our context of classical approximability:
\begin{theorem}[\cite{van2}]\label{van2-2}
Let $U$ be an arbitrary unitary operation on $n$ qubits such that $U|0^n\rangle$ is CT, and $O$ be an arbitrary Hermitian operation 
with $||O|| \leq 1$, where $||O||$ is the maximum of the absolute values of the eigenvalues 
of $O$. Let $V$ be an arbitrary unitary operation on $n$ qubits such that $V^\dag OV$ is ECS. Then, 
the real number
$$\langle 0^n|U^\dag V^\dag O VU|0^n\rangle$$
is classically approximable.
\end{theorem}

\subsection{IQP circuits and the Quantum part of Simon's algorithm}

We consider the following quantum circuit on $n$ qubits initialized to $|0^n\rangle$~\cite{van2}:
\begin{enumerate}
\item Apply Hadamard gates to some subset of qubits $Q$.

\item Apply a polynomial-size quantum circuit $D_n$ on $n$ qubits for an efficiently computable basis-preserving operation.

\item Apply Hadamard gates to some subset of qubits $R$.
\end{enumerate}
The special case $Q=R=\{1,\ldots ,n\}$ is depicted in Fig.~\ref{figure2} (a). For example, when $Q=R=\{1,\ldots ,n\}$, $m=n$, and 
$D_n$ consists of $Z$-diagonal gates such as $Z$, $CZ$, and $CCZ$ gates, the circuit reduces to an IQP circuit. IQP circuits 
have been extensively studied as candidates for demonstrating quantum computational advantage. Moreover, the quantum 
part of Simon's algorithm is the one with $Q=R=\{1,\ldots ,n/2\}$ and 
$m=n/2$\footnote{Precisely speaking, the basis-preserving operation in Simon's algorithm depends on a given oracle and thus 
we need to consider an argument relative to the oracle, but it is straightforward to do this~\cite{Schwarz}.}.

Let $C_n$ be an arbitrary polynomial-size quantum circuit of this form, i.e., $C_n=H_RD_nH_Q$, where $H_Q$ and $H_R$ are layers of $H$ gates 
on qubits $Q$ and on qubits $R$, respectively. We can show that, for any $s \in \{0,1\}^m\setminus\{0^m\}$,
$$\langle 0^n |C_n^\dagger (Z(s) \otimes I_{n-m}) C_n |0^n\rangle = \langle 0^n |H_QD_n^\dagger H_R (Z(s) \otimes I_{n-m}) H_RD_nH_Q |0^n\rangle$$
is classically approximable. Indeed, the product state $H_Q|0^n\rangle$ is a CT state and thus, by Theorem~\ref{van2-1}, $D_nH_Q|0^n\rangle$ is also a CT state. 
Moreover, since $HZH=X$ and $H^2=I$, the operation $H_R (Z(s) \otimes I_{n-m}) H_R$ can be represented as a tensor product of $X$ and $I$, which is obviously 
an ECS operation. Thus, by Theorem~\ref{van2-2}, the Pauli expectation value
$$\langle 0^n |C_n^\dagger (Z(s) \otimes I_{n-m}) C_n |0^n\rangle$$
is classically approximable. Therefore, by Theorem~\ref{main-theorem}, for any SCP $f_m$, $C_n$ followed by $f_m$ is classically simulable. 
This recovers the result of Van den Nest~\cite{van2}.

\subsection{Clifford Magic circuits}

We consider a Clifford Magic circuit on $n$ qubits initialized to $|0^n\rangle$~\cite{Yoganathan}, which has been proposed as 
another model exhibiting quantum computational advantage:
\begin{enumerate}
\item Apply Hadamard gates to all qubits.

\item Apply $T$ gates to all qubits.

\item Apply a polynomial-size Clifford circuit $E_n$ on $n$ qubits, which consists of $H$ gates, $S$ gates, and $CZ$ gates.
\end{enumerate}
This circuit is depicted in Fig.~\ref{figure2} (b), where $m$ is usually set to $n$ when the classical hardness of sampling its output probability distribution is discussed.

Let $C_n=E_nT^{\otimes n}H^{\otimes n}$ be an arbitrary polynomial-size Clifford Magic circuit. We can show that, for any $s \in \{0,1\}^m\setminus\{0^m\}$,
$$\langle 0^n |C_n^\dagger (Z(s) \otimes I_{n-m}) C_n |0^n\rangle 
= \langle 0^n |H^{\otimes n}  T^{\dagger \otimes n}  E_n^\dagger(Z(s) \otimes I_{n-m}) E_nT^{\otimes n}H^{\otimes n} |0^n\rangle$$
is classically approximable. Indeed, $T^{\otimes n}H^{\otimes n}|0^n\rangle$ is a product state and thus a CT state. Moreover, since $E_n$ is a Clifford circuit, 
the operation $E_n^\dagger (Z(s) \otimes I_{n-m}) E_n$ is a Pauli operation, which is obviously an ECS operation. Thus, by Theorem~\ref{van2-2}, 
the Pauli expectation value
$$\langle 0^n |C_n^\dagger (Z(s) \otimes I_{n-m}) C_n |0^n\rangle$$
is classically approximable. Therefore, by Theorem~\ref{main-theorem}, for any SCP $f_m$, $C_n$ followed by $f_m$ is classically simulable. 
It is worth noting that this result does not follow immediately from the original framework of Van den Nest~\cite{van2}, which 
highlights the broader applicability of our Theorem~\ref{main-theorem} to circuits involving magic states.

\begin{figure}[t]
\centering
\includegraphics[width=1.0\linewidth]{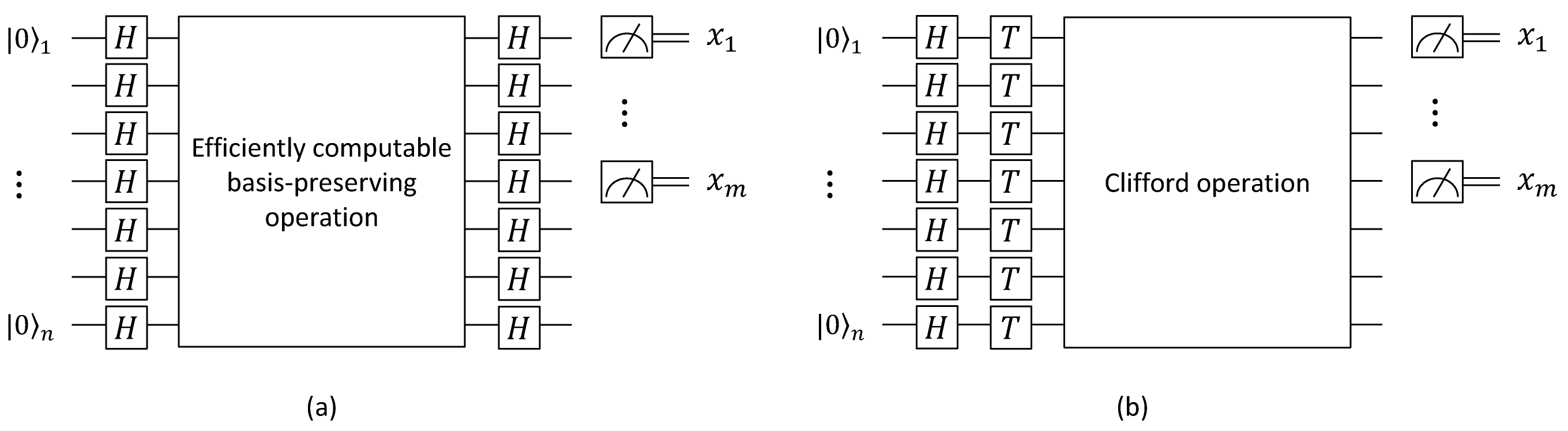}
\caption{(a) Quantum circuit for an efficiently computable basis-preserving operation sandwiched between two Hadamard layers. 
Examples of such circuits are IQP circuits and the quantum part of Simon's algorithm. (b) Clifford Magic circuit.}
\label{figure2}
\end{figure}

\section{Constant-depth quantum circuits followed by SCP}\label{commuting-section}

To analyze the simulablity of constant-depth quantum circuits followed by SCP, we introduce the model of commuting quantum circuits. 
A commuting quantum circuit on $n$ qubits is a quantum circuit on $n$ qubits consisting of pairwise 
commuting gates, where each commuting gate acts non-trivially on a constant number of qubits. A commuting quantum circuit on $n$ qubits is 
said to be $k$-local if the support of each commuting gate contains at most $k$ qubits, where $k=O(1)$~\cite{Ni}.

Let $C_n$ be an arbitrary polynomial-size quantum circuit of constant depth $d$ on $n$ qubits initialized to $|0^n\rangle$, and 
$m$ be an arbitrary integer such that $m \leq n \leq {\rm poly}(m)$. Moreover, we fix an arbitrary SCP $f_m$ and consider 
$g_m(x)=(-1)^{f_m(x)}$. To show Theorem~\ref{second-theorem}, we apply the algorithm described immediately before the 
proof of Lemma~\ref{key-lemma}. More concretely, it suffices to show that, for each $s \in \widetilde{L}$, the Pauli expectation value 
$$\langle 0^n |C_n^\dagger (Z(s) \otimes I_{n-m}) C_n |0^n\rangle$$
is approximated by a polynomial-time probabilistic algorithm with access to commuting quantum circuits, 
which satisfy the conditions in Theorem~\ref{second-theorem}, with polynomially-small additive error 
and with exponentially-small error probability. Since the Pauli expectation value is 1 if $s=0^m$, we set its estimate to 1 and 
focus on the case $s\neq 0^m$ in the following. We employ the Hadamard test, which allows us to estimate 
the desired Pauli expectation value. Indeed, for an arbitrary fixed $s \in \widetilde{L}$, we use the following quantum circuit 
on $n+1$ qubits initialized to $|0^{n+1}\rangle$:
\begin{enumerate}
\item Apply a Hadamard gate to the top qubit and $C_n$ to the other $n$ qubits.

\item For each $1 \leq j \leq m$ with $s_j=1$, apply a controlled-$Z$ gate to the pair of the top qubit and qubit $j$.

\item Apply a Hadamard gate to the top qubit and $C_n^\dagger$ to the other $n$ qubits.
\end{enumerate}
The circuit with $m=n$ and $s=1^n$ is depicted in the left hand side of Fig.~\ref{figure3}. The circuit for the Hadamard test is related to 
the Pauli expectation value as follows:
\begin{lemma}\label{out-prob}
When the top qubit of the quantum circuit for the Hadamard test is measured in the $Z$-basis after Step 3, the probability of obtaining $0$ is
$${\rm Prob}(0)= \frac{1+ \langle 0^n |C_n^\dagger (Z(s) \otimes I_{n-m}) C_n |0^n\rangle }{2}.$$
\end{lemma}
\begin{proof}
By the quantum circuit on $n+1$ qubits described above, we obtain the following state:
\begin{align*}
|0^{n+1}\rangle & \mapsto  \frac{1}{\sqrt{2}}|0\rangle\otimes C_n|0^n\rangle + \frac{1}{\sqrt{2}}|1\rangle\otimes C_n|0^n\rangle \\
                     & \mapsto  \frac{1}{\sqrt{2}}|0\rangle\otimes C_n|0^n\rangle + \frac{1}{\sqrt{2}}|1\rangle\otimes  (Z(s) \otimes I_{n-m}) C_n|0^n\rangle \\
                     & \mapsto  \frac{1}{\sqrt{2}}|0\rangle \otimes \frac{I+C_n^\dagger  (Z(s) \otimes I_{n-m}) C_n}{\sqrt{2}}|0^n\rangle 
+ \frac{1}{\sqrt{2}}|1\rangle \otimes \frac{I - C_n^\dagger  (Z(s) \otimes I_{n-m}) C_n}{\sqrt{2}}|0^n\rangle,
\end{align*}
where each arrow corresponds to each step. Thus,
\begin{align*}
{\rm Prob}(0) &= \frac{1}{4}\langle 0^n|(I+C_n^\dagger (Z(s) \otimes I_{n-m}) C_n)(I+C_n^\dagger (Z(s) \otimes I_{n-m}) C_n) |0^n\rangle \\
&= \frac{1+ \langle 0^n |C_n^\dagger (Z(s) \otimes I_{n-m}) C_n |0^n\rangle }{2},
\end{align*}
which is the desired relationship.
\end{proof}

Using this lemma, we can show Theorem~\ref{second-theorem} as follows:
\begin{proof}[Proof of Theorem~\ref{second-theorem}]
Let $X$ be the real-valued random variable defined as the classical outcome of the measurement in Lemma~\ref{out-prob}. Then,
$${\mathbb E}[X] =  \frac{1- \langle 0^n |C_n^\dagger (Z(s) \otimes I_{n-m}) C_n |0^n\rangle }{2},$$
which is the probability of obtaining 1. Thus, by the Chernoff-Hoeffding bound as described in Section~\ref{Chernoff-section}, there exists a polynomial-time probabilistic 
algorithm with access to the above quantum circuit that approximates the desired value
$$\langle 0^n |C_n^\dagger (Z(s) \otimes I_{n-m}) C_n |0^n\rangle$$
with polynomially-small additive error and with exponentially-small error probability.

\begin{figure}[t]
\centering
\includegraphics[width=1.0\linewidth]{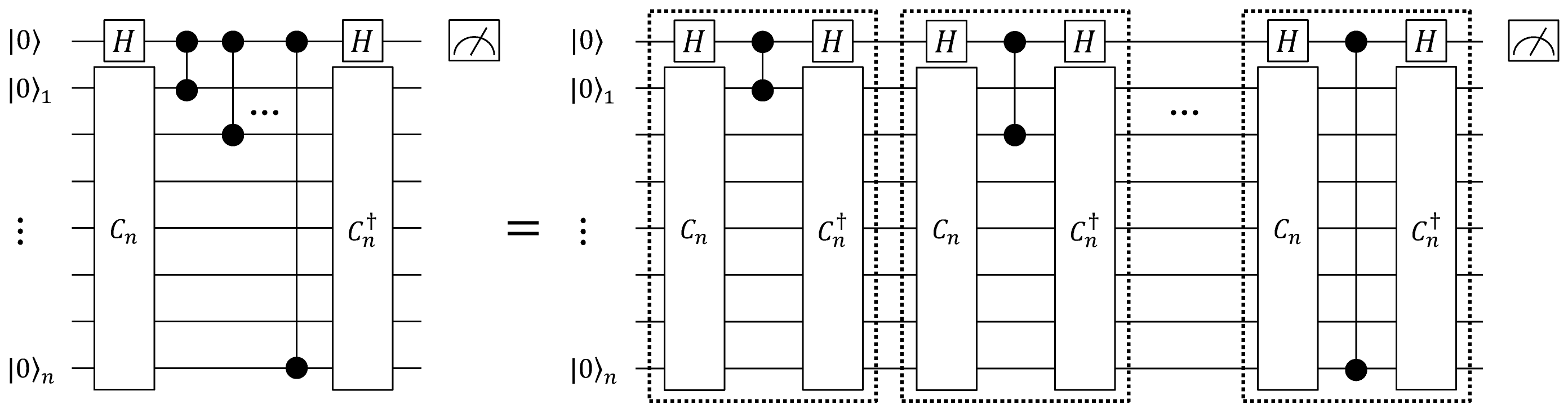}
\caption{Quantum circuit for approximating $\langle 0^n |C_n^\dagger Z(1^n) C_n |0^n\rangle$. 
The right hand side of the equation is a commuting quantum circuit whose output probability distribution is the same as that of the left hand side.}
\label{figure3}
\end{figure}

We now transform the above quantum circuit for the Hadamard test on $n+1$ qubits into a commuting quantum circuit 
on $n+1$ qubits without affecting the output probability distribution. Following the approach of~\cite{Ni}, this can be done by 
using the relationship
$$ C_n^\dagger (Z(s) \otimes I_{n-m}) C_n  =  \prod_{1 \leq j \leq m:s_j=1} (C_n^\dagger (Z_j \otimes I_{n-m}) C_n).$$
The resulting commuting quantum circuit with $m=n$ and $s=1^n$ is depicted in the right hand side of Fig.~\ref{figure3}. Indeed, 
for any $1 \leq j \leq m$ with $s_j=1$, we regard 
$$(H\otimes C_n^\dagger) (CZ_j \otimes I_{n-m}) (H\otimes C_n)$$
as a gate, where $CZ_j$ denotes the controlled-$Z_j$ gate acting on the top qubit and qubit $j$. 
The number of commuting gates in the commuting quantum circuit is $|s|$, which is upper-bounded by $\deg(g_m)$. Since $g_m(x)=1-2f_m(x)$ and $f_m$ is a non-zero function, 
it holds that $\deg(g_m) = \deg(f_m)$. Moreover, the number of qubits on which each commuting gate acts is at most
$$1+ \max \{|{\rm supp}(C_n^\dagger (Z_j \otimes I_{n-m}) C_n)| \mid 1 \leq j \leq m \}.$$
This completes the proof of Theorem~\ref{second-theorem}.
\end{proof}

\section{Conclusions and future work}

We studied the classical simulability of a polynomial-size quantum circuit $C_n$ followed by SCP $f_m$,
where $m \le n \le \mathrm{poly}(m)$. Our first contribution is a necessary and sufficient condition on $C_n$ ensuring that, for any SCP $f_m$, 
$C_n$ followed by $f_m$ is classically simulable. This characterization recovers the result of Van den Nest and applies equally to IQP circuits 
and Clifford Magic circuits, the latter of which had not been addressed in prior work.

We then considered the case where $C_n$ has constant depth. Although it is unlikely that for any SCP $f_m$, $C_n$ followed by $f_m$ is 
classically simulable, we showed that it becomes simulable by a polynomial-time probabilistic algorithm with access to commuting 
quantum circuits on $n+1$ qubits. We further quantified the required quantum resources by bounding both the number of commuting gates and 
their locality in terms of $\deg(f_m)$ and $|{\rm supp}(C_n^\dagger (Z_j \otimes I_{n-m}) C_n)|$. This provides a refined understanding of the 
hardness of simulating constant-depth quantum circuits followed by SCP.

Our results suggest several promising avenues for further research.

\begin{itemize}
    \item \textbf{Beyond Fourier sparsity.}
    SCP is defined via Fourier sparsity, but it would be natural to investigate whether our techniques extend to broader classes of post-processing functions, such as functions with quasi-polynomial Fourier sparsity or approximate sparsity. Understanding the precise boundary of post-processing that preserves classical simulability remains an open challenge.

    \item \textbf{Sharper bounds for constant-depth quantum circuits.}
    For constant-depth circuits, it would be interesting to determine whether the commuting quantum circuits used in Theorem~\ref{second-theorem} are optimal, or whether tighter upper or lower bounds on the required quantum resources can be established. In particular, understanding the minimal locality or depth needed for such commuting circuits could lead to a finer-grained hierarchy of simulability.

    \item \textbf{Computational power of commuting quantum circuits.}
    Our analysis highlights the utility of commuting quantum circuits as a simulation primitive. Further exploring their computational power---especially in relation to non-commuting extensions, Clifford commuting circuits, or restricted gate sets---may yield new insights into the structure of quantum advantage~\cite{Ni,Takahashi-commuting}.

\item \textbf{Functional sparsity vs.\ distributional sparsity.}
Our work focuses on functional sparsity in the sense of Fourier sparsity, whereas Schwarz et al.~\cite{Schwarz} studied distributional sparsity, where only polynomially many outcomes of the output distribution have non-negligible probability. Clarifying the relationship between these two notions may lead to a unified understanding of when quantum circuits admit efficient classical simulation.

    \item \textbf{Other quantum circuit families.}
    It would also be valuable to examine the classical simulability of other circuit families believed to perform classically-hard tasks, such as conjugated Clifford circuits~\cite{Bouland}. Applying our framework to these models may reveal new structural properties or unexpected simulability regions.
\end{itemize}

Overall, our results contribute to the broader program of understanding the boundary between quantum and classical computation, and we hope they stimulate further investigation into the interplay between quantum circuit structure, classical post-processing, and classical simulability.

\section*{Acknowledgements}
This work was supported by JST K Program Grant Number JPMJKP24U2, Japan. 



\bibliography{mybib-TCS}

\end{document}